\documentclass[conference, anonymous]{IEEEtran}
\usepackage{cite}
\ifCLASSINFOpdf
   \usepackage[pdftex]{graphicx}
\else
   \usepackage[dvips]{graphicx}
\fi
\usepackage[cmex10]{amsmath}
\usepackage{algorithmic}
\usepackage{algorithm}

\usepackage{array}

\usepackage{mdwmath}
\usepackage{mdwtab}

\usepackage{eqparbox}

\usepackage{fixltx2e}

\usepackage{stfloats}

\usepackage{graphicx}
\usepackage{caption}
\usepackage{subcaption}
\usepackage{url}

\newtheorem{theorem}{Theorem}
\newtheorem{lemma}{Lemma}

\begin{document}
%

\title{Optimal Covering with Mobile Sensors in an Unbounded Region}
\author{\IEEEauthorblockN{Barun Gorain and Partha Sarathi Mandal}
\IEEEauthorblockA{Department of Mathematics\\
Indian Institute of Technology Guwahati, Guwahati-781039, India\\
Email: \{b.gorain,psm\}@iitg.ernet.in}}

\maketitle
\begin{abstract}
Covering a bounded region with minimum number of homogeneous sensor nodes is a NP-complete problem \cite{Li09}.
In this paper we have proposed an {\it id} based distributed algorithm for optimal coverage
in an unbounded region. The proposed algorithm guarantees maximum spreading in $O(\sqrt{n})$ rounds
without creating any coverage hole. The algorithm executes in synchronous rounds without exchanging any message.

We have also explained how our proposed algorithm can achieve optimal energy consumption and handle random
sensor node deployment for optimal spreading.
\end{abstract}


\IEEEpeerreviewmaketitle

\section{Introduction}

Coverage with mobile sensor nodes in wireless sensor networks (WSNs) is one of the most important
issue in modern time. In general coverage is defined as the measurement of the quality of surveillance
of sensing function for a sensor network. In practice, sensor nodes are deployed in a region of interest
to get desire information from the region. Amount of information collected from the region of interest
depends on how well the target region is covered by the deployed sensor nodes.
For example, in the case of forest monitoring \cite{Turjman09,Tao09}, every location of the forest must be covered by
at least one sensor node such that if fire starts in a specific location it can be immediately detected.
The objective of the coverage problem is to improve the coverage performance when the WSN
is unable to satisfy the requirements and if it satisfies the requirement, then how to extend the network
lifetime with coverage guarantee. Coverage guarantee may be possible using controlled placement of sensor nodes
on a target region which are focussed on the papers \cite{Brooks04,Krishnamurthy05}. Deployment of sensor nodes
on a remote area is also a challenging problem.
Controlled placement is not always possible on any remote area. Random deployment of sensor nodes may be
possible using some airdrop \cite{Ishizuka04,Tang06} on a remote region of interest. There are many works
\cite{Brooks04,Ishizuka04,Krishnamurthy05,Tang06,Younis08} proposed in literature on node-deployment strategy.
Random deployment of static sensor nodes on a target region may not guarantee complete coverage.
Introducing limited mobility over the static sensor nodes, it is possible to improve the coverage by
reducing overlaps with neighboring nodes and allowing them to move towards the uncovered region.

In wireless mobile sensor networks, several heuristic algorithms have been proposed in literature to
cover a bounded region with mobile sensor nodes \cite {Ma,Saravi09,WangCP03,Wang}. But none of the
approaches discussed about the optimal coverage or maximum spreading for an unbounded region with
a given set of mobile sensor nodes.
In this proposed work we have designed a distributed  algorithm for optimal coverage
(maximum spreading) in an unbounded region with a given set of mobile sensor nodes.
Initially all nodes are deployed in a two dimensional plane, after that the nodes move and place
themselves at the vertices of an equilateral triangular tiling to achieve maximum spreading \cite {Bartolini10,Ma}.
The condition for the spreading is that there is no uncover region surrounded by the sensor nodes.
We use the term `coverage hole' or simple `hole' for uncover region and `node' for sensor node
throughout the paper.

\subsection{Our Contribution}
In this work we have proposed a distributed synchronous algorithm for optimal coverage
{\it i.e.}, maximum spreading on an unbounded region for a given number of mobile sensor nodes.
The proposed algorithm is {\it id} based, no message exchange is required for the
execution of the algorithm. The algorithm is executed in synchronous rounds and $O(\sqrt{n})$ rounds are
required for maximum spreading, where $n$ is the total number of sensor nodes.
Our proposed algorithm can achieve optimal energy consumption if each node first computes final
destination for optimal spreading and then moves straight to the destination.
The algorithm can handle random node deployment instead of strategic deployment at a point
on the plane.

Bartolini {\it et al.} in \cite{Bartolini10} proposed a `Push \& Pull' algorithm to complete uniform
coverage on a bounded region where multiple time message broadcasting is required to coordinate
their movements. Objective of the node movement is to form hexagonal tiling, which is similar
to our approach. Though the authors proved finite time convergency of their algorithm to complete
uniform coverage but they have not given the complexity measure in terms of number of the sensor node
for termination of the algorithm. In contrary to the above approach, our algorithm forms
an equilateral triangular tiling in $\left\lceil-\frac{1}{2}+\sqrt{\frac{4n-1}{3}}~\right\rceil$
rounds, which is the optimal coverage for given set of $n$ nodes. With respect to the energy
consumption perspective our approach outperforms than the previous one because no message exchange
is required.

\section{Related Works}
Several approaches has been found in the literature to overcome coverage problem in mobile wireless
sensor networks. Wang {\it et al.} in the paper \cite{Wang} proposed three movement assisted sensor
deployment algorithms,  which are  VEC (vector based algorithm), VOR (voronoi based algorithm) and Minimax.
In this paper the voronoi diagram is used to identify coverage holes. Then in VEC algorithm,
the sensor nodes which not cover their respective voronoi cell are pushed to fills the coverage gaps.
The sensor nodes are moved to the farthest voronoi vertex in VOR and minimize the coverage hole.
Minimax algorithm moves the sensor nodes to close to the farthest voronoi vertex avoiding the
generation of new holes.

There is another approach proposed by Cao {\it et al.} in \cite{WangCP03} which is based on voronoi diagram.
A network of static and mobile sensor nodes are considered in this work. Firstly, the static and mobile
both sensor nodes are randomly deployed over a region and the static nodes find the presence of hole by
computing their local voronoi cell. Then the static nodes on the boundary of holes bids for mobile sensor
nodes to move to the farthest voronoi point. The authors proposed two kinds of bidding protocol, one is
distance based and other is price based. In the distance based protocol the static nodes on hole boundary
bid for the mobile nodes which is closest to the corresponding static node. In the price based protocol,
the mobile nodes bidden is the cheapest one with respect to coverage improvement. Mobile nodes can move
one hole to another hole if the coverage is improved compare to the prior movement. Next another proxy
based protocol is defined where nodes will first calculate their final location and will move to their
target location only once.

The ATRI algorithm proposed by Ma {\it et al.} \cite{Ma} makes the overall deployment layout close to equilateral
triangulation by moving sensor nodes after initial random deployment. Here if entire monitoring area is
triangulated by delaunay triangulation and the triangles are equilateral with edge length $\sqrt{3}r$,
then the coverage area of $n$ nodes is maximum without coverage hole. Where $r$ is the sensing radius of
the sensor nodes. In their proposed algorithm they assume completely connected network and they divide the
sensing disk of every sensor nodes in to six equal sectors. In every iteration of their algorithm, each node
wants to adjust its position with the nearest neighbor for each sector in a distance $\sqrt{3}r$ and calculate
movement vector. Adding the movement vectors of the corresponding six sectors, the sensor nodes moves in the
direction of the resultant vector. After every iteration the layout of the network will become closer to the
optimum layout.

One centralized approach of sensor movement is proposed by Saravi {\it et al.} in \cite {Saravi09}. In this approach the
given area is delaunay triangulated where each triangles are equilateral with side $\sqrt{3}r$. Each of the sensor
nodes has information of the delaunay vertices of the triangulation. After random deployment of the sensor nodes,
each nodes will move to their closest delaunay vertex.

Bartolini {\it et al.} in \cite {Bartolini10}, proposed an algorithm: Push \& Pull for complete coverage in a bounded region. In this approach sensor nodes make an equilateral triangular tiling on a plane by their movements.
A regular hexagonal structure is formed by an initiating node, located at center with an arbitrary choice of six neighbouring nodes and their appropriate movements. A node whose hexagonal structure is already formed is called a snapped otherwise unsnapped node. If some unsnapped sensor nodes are located inside the hexagon of a snapped sensor nodes, then the snapped node Push these unsnapped nodes to the lower density area of the plane. If some snapped nodes detect any coverage hole adjacent to their hexagon, they send hole trigger messages for attracting unsnapped nodes for filling the coverage hole. The process of message sending and attracting nodes is called Pull activity. If two different clusters with their tiling have different orientation then tiling marge activity is applied to marge into a common tiling.

The rest of the paper is organized as follows: In section \ref{bas} we discuss the formal description and theoretical analysis of our approach. In section \ref{cover} the underlying system model and the proposed algorithm is described and finally section \ref{concl} conclude the paper.

\section{Basic Idea}
\label{bas}
Covering a bounded region with minimum number of homogeneous sensor nodes is a NP-complete problem \cite{Li09}.
But when the target region is unbounded, optimal coverage {\it i.e.}, maximum spreading with a given number of homogeneous sensor nodes can be achieved if sensor nodes are placed at the vertices of a triangular lattice
arrangement \cite{Bartolini10,Brass:2007,Ma,Wang11} with edge length of triangle is $\sqrt{3}r$. Where $r$
is the sensing range of each sensor node. The hexagonal tiling corresponding to a triangular lattice arrangement gives the optimal coverage and density as explain in the paper \cite{Bartolini10}. After deployment of a set of mobile sensor nodes, if all nodes move to all vertices of some triangles of
such triangular lattice arrangement then the total coverage is optimal \cite{Ma}.
\begin{figure}[h]
  \centering
  \includegraphics[width=0.2\textwidth]{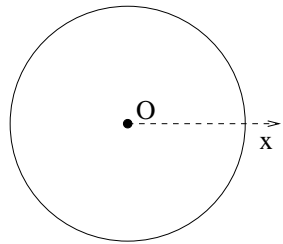}
  \caption{Showing the initial deployment at $O$}
  \label{fig:point1}
\end{figure}

In our approach, there are two {\it status} of each sensor node; {\it stable} and {\it unstable}. In every round some of the  nodes become {\it stable} at their desire locations. All {\it stable} nodes stop execution further for their movements. The remanning nodes are {\it unstable} in the round.  Every {\it unstable} nodes keep moving in every round until they become {\it stable} at their desire locations.
For simplicity, in our proposed algorithm we assume point sensor, {\it i.e.}, a geometric point on a plane but same technique is also applicable for non-point sensors too, which has explained in section \ref{concl}. Initially, $n$ nodes with $id$ $\{0, 1, 2, \cdots, n-1\}$ are deployed at a point say, $O$ as shown in the Fig. \ref{fig:point1}, where $O$ is considered as origin.
The aim of our approach is to give a movement strategy such that every node can moves to a vertex of a hexagonal tiling of the plane.
In each step of the algorithm, the unstable nodes move to the positions for which the hexagon corresponding to every stable node is completed.
In the initial round, the node with {\it id} 0 becomes stable in its initial position and all other nodes are decomposed into six different subgroups and move to six different locations to complete the hexagonal structure corresponding to the position of node with {\it id} $0$. The initial deployment of all nodes and the direction of movements in this round is shown in the Fig. \ref{fig:point1} and Fig. \ref{fig:point2} respectively. The Fig. \ref{fig:point3} shows the  equilateral triangular tiling after execution of the initial round.

\begin{figure}[h]
  \centering
  \includegraphics[width=0.3\textwidth]{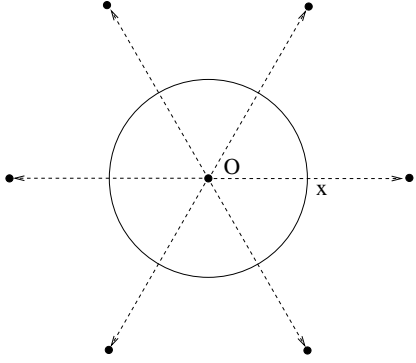}
  \caption{The directions to move in round 0}
  \label{fig:point2}
\end{figure}
 \begin{figure}[h]
  \centering
  \includegraphics[width=0.3\textwidth]{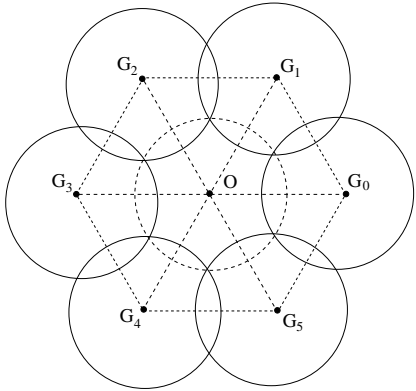}
  \caption{Formation of equilateral triangular tiling at the end of round 0}
  \label{fig:point3}
\end{figure}

For a node with $id$ $i$, the variable $(X_i,Y_i)$ stores its current position and $\theta_i$ stores the angle
between the direction of its movement and the positive $x-$axis. In every round, all unstable nodes update their positions $(X_i,Y_i)$  to their next positions and then move to the updated positions. There is another variable {\it type} for each unstable nodes, which is either {\it single} or {\it double}.
All the unstable nodes with {\it type single} update their positions using the Eqn. \ref{eq:eq1}.
These {\it type} of nodes do not change their direction of movements till reach to the {\it stable} state.
\begin{eqnarray}\label{eq:eq1}
X_i&=&X_i+\sqrt{3}r\cos\theta_i \nonumber \\
Y_i&=&Y_i+\sqrt{3}r\sin\theta_i
\end{eqnarray}

The nodes with {\it type double} located at the same position in a round, decompose into two subgroups
based on the value $val$, $val= -\frac{1}{2}+\sqrt{\left(\frac{m\_id_i-j}{3}+\frac{1}{4}\right)} $, where $m\_id_i$ and $j$ are explained in section \ref{sysM}. Then the node updates its next position using the Eqn. \ref{eq:eq1} if $val$ is not an integer, otherwise updates its next position using the Eqn. \ref{eq:eq2} and also updates its {\it type} as {\it single}.
In the outer layer of the Fig. \ref{fig:round3}, one and two directions are shown for the nodes with {\it type single} and {\it type double} respectively.

\begin{eqnarray}\label{eq:eq2}
X_i&=&X_i+\sqrt{3}r\cos\left(\theta_i+\frac{\pi}{3}\right)\nonumber \\
Y_i&=&Y_i+\sqrt{3}r\sin\left(\theta_i+\frac{\pi}{3}\right)
\end{eqnarray}
The execution of round 1 is shown in the Fig. \ref{fig:point4} and Fig. \ref{fig:point5}.

\begin{figure}[h]
  \centering
  \includegraphics[width=0.4\textwidth]{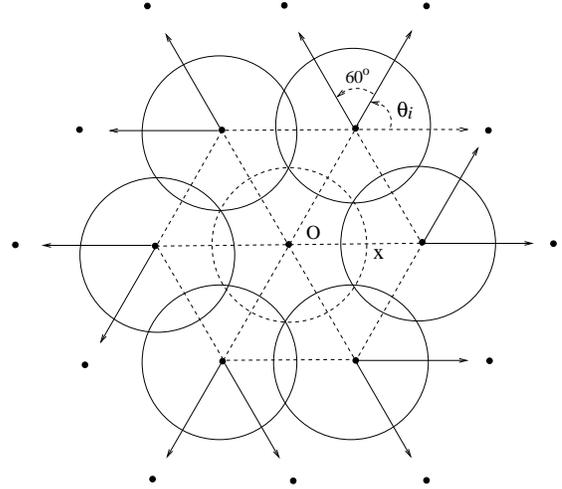}
  \caption{The locations and the corresponding directions to move in round 1 from the outer layer locations of round 0}
 \label{fig:point4}
\end{figure}
\begin{figure}[h]
  \centering
  \includegraphics[width=0.3\textwidth]{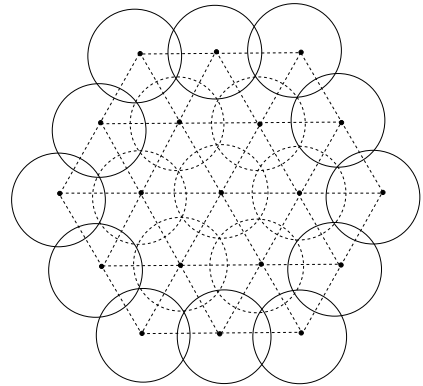}
  \caption{Formation of hexagonal structures surrounded by each location of previous round 0 at the end of the round 1}
  \label{fig:point5}
\end{figure}

\begin{lemma}\label{lem:double}
There are exactly six different groups of nodes with {\it type double} in each round of execution of the algorithm \textsc{MaxCover}.
\end{lemma}
\begin{proof}
In the initial round or round 0, all unstable nodes are decomposed into six different groups (refer steps \ref{initialStart}-\ref{initialEnd} of the algorithm \textsc{MaxCover}). The {\it type} of the nodes of every group are {\it double} at the end of this round. But in every round each group of nodes of {\it type} double decomposed into two subgroups according to the steps \ref{decomposeStart} - \ref{decomposeEnd} of the algorithm based on their ids. During this decomposition all nodes in a subgroup change their {\it type} to {\it single} and the {\it type} of the nodes in other subgroup remain {\it double}. So, total number of group with  {\it type double} remain unchanged in every round of execution of the proposed algorithm \textsc{MaxCover}.
\end{proof}

The Fig. \ref{fig:round3} and Fig. \ref{fig:round4} are showing the execution of the algorithm in round 2, where all the nodes with {\it type double} move in two different directions and the nodes with  {\it type single} move in the same direction as in the previous round.

\begin{lemma}\label{lem1:stable}
There are $6k$ number of nodes become {\it stable} in round $k$ for $k\ge1$.
\end{lemma}

\begin{proof}
We use induction on the number of rounds to prove this Lemma.
At the initial round of the algorithm, all unstable nodes are decomposed into six different groups and the nodes with id $1$ to $6$ are assign as $min\_id$ of these groups (refer step \ref{setId1} and step \ref{setId2} of the algorithm \textsc{MaxCover}). All the nodes with $id = min\_id$ becomes {\it stable} in the next round, {\it i.e.}, in round $1$. Therefore $6 \times 1$  nodes become {\it stable} in round $1$. Hence the lemma is true for $k=1$.

Suppose the statement of this lemma is true for round $m$, {\it i.e.}, $6m$ nodes become {\it stable} in round $m$. Now, let us consider the execution of the algorithm in the round $m$. There are $6m$ groups in round $m$ since in every round of the algorithm, one node from each group become {\it stable}.
According to the Lemma \ref{lem:double} six groups out of these $6m$ groups are with node's {\it type double} and the remaining $6m-6$ groups are with node's {\it type single}. All the groups with node's {\it type double} decompose into two subgroups and hence total twelve subgroups are formed. So total number of groups including
nodes with {\it type single} and {\it double} are $6m-6+12$ at the end of the round $m$. Hence there are $6(m+1)$ numbers of groups at the beginning of the round $m+1$.  As one node from each group become {\it stable} total $6(m+1)$ nodes become {\it stable} in the round $m+1$. Therefore by the mathematical induction, the lemma follows.
\end{proof}

\begin{figure}[h]
  \centering
  \includegraphics[width=0.35\textwidth]{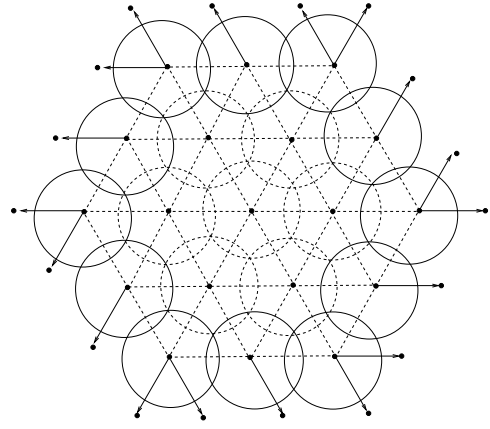}
  \caption{The locations and corresponding directions to move in round 2 from the outer layer locations of round 1}
 \label{fig:round3}
\end{figure}

\begin{figure}[h]
  \centering
  \includegraphics[width=0.35\textwidth]{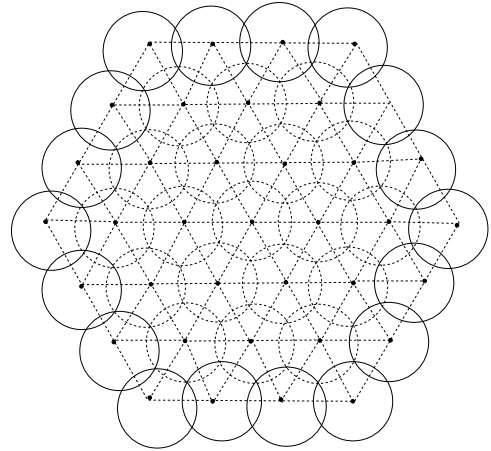}
  \caption{Showing the hexagonal tiling at the end of round 2}
 \label{fig:round4}
\end{figure}

\begin{lemma}\label{lem2:total}
After $\left\lceil-\frac{1}{2}+\sqrt{\frac{4n-1}{3}}~\right\rceil$ rounds of execution, all nodes become {\it stable}.
\end{lemma}
\begin{proof}
It is clear from the Lemma \ref{lem1:stable} that in every round $k$ ($\ge1$), $6k$ nodes becomes {\it stable} and the node with {\it id} $0$ becomes stable in the initial round. Let $m$ be the number of rounds required to {\it stable} $n$ nodes. Then
\begin{eqnarray}\label{eq:eq3}
    n &\le& 1+6\times 1+6\times 2+\cdots+6m\nonumber\\
       &=&1+3m(m+1)
\end{eqnarray}

Solving the above inequality  \ref{eq:eq3} we can get number of round $m$ in terms of total number of nodes $n$,  $m=\left\lceil-\frac{1}{2}+\sqrt{\frac{4n-1}{3}}~\right\rceil$.
\end{proof}
\begin{lemma}\label{lem3:tiling}
When all nodes become {\it stable} then their positions form an equilateral triangular tiling.
\end{lemma}
\begin{proof}
We prove this Lemma by mathematical induction on the number of rounds.

{\it Base Step:} In the initial round (round 0), all the nodes except node with {\it id} 0 move to six different
positions such that these coordinates of the positions make six equilateral triangles with a common
vertex $O$ as shown in the Fig. \ref{fig:point3}. Which is same as a hexagonal
structure centering at $O$. So, at the end of round $0$ ($k=0$), the positions of the sensor nodes
corresponds an equilateral triangular tiling.

{\it Inductive Step:} Let the positions of the sensor nodes makes an equilateral triangular
tiling after round $k$. One node in each group of {\it unstable} nodes become stable and other nodes
move to their respective neighboring vertices.
The coordinates of the neighboring vertices are computed using the equations
\ref{eq:eq1} and \ref{eq:eq2} depending on the {\it type} of the nodes as described above
and move to complete the hexagonal structure centering at location of each group in the round $k$.
All new location which are computed in round $k+1$ add a new layer and form an equilateral triangular
tiling over the existing structure at round $k$. Hence by induction, the lemma follows.
\end{proof}
\begin{theorem}
After $O(\sqrt{n})$ rounds of execution, all sensor nodes become {\it stable} with optimal spreading by $n$ nodes.
\end{theorem}
\begin{proof}
Theorem follows from the Lemma \ref{lem2:total} and \ref{lem3:tiling}.
\end{proof}

\section{Proposed Covering Algorithm}
\label{cover}

\subsection{System Model}
\label{sysM}
Sensor nodes are homogeneous, {\it i.e.}, each node have same communication and sensing range. The sensing range is $r$ for each sensor node.
We have considered mobile sensor nodes and for simplicity we assume that the sensor nodes are
geometric points on a two dimensional plane, which can move freely over the plane. Each node has an unique {\it id} $\in \{0,1,2,\cdots,n-1\}$, where $n$ is the total number of sensor nodes. A node can compute coordinates locally based on a common coordinate system and can move to the location.

Each node $i$ maintains the variables: {\it status}, $type$, $min\_id$, $m\_id_i$, $\theta_i$, $(X_i,Y_i)$.
The {\it status} may be {\it stable} or {\it unstable}. The $status$ stable means that the node has placed itself at a right location, do not require to execute the algorithm further. Initial status of all nodes are unstable. The $status$ unstable means that the node is moving towards its right location except the node with {\it id} 0. In initial round, node with {\it id} 0 changes its {\it status} to stable without any movement.  The variable $type$ may either single or double.
A node with $type$ double has two possible direction of movement based on its $id$ but for the node with $type$ single has the same direction of its movement as the previous round.
The $min\_id$ is the minimum {\it id} among all unstable nodes in a group. The $m\_id_i$ is updated in each round for a node with $type$ double.
The $\theta_i$ is the angle made by the direction to move in a round with the positive direction of $x$-axis. The $\theta_i$ is updated in each round. The coordinates $(X_i,Y_i)$ is the position of node $i$, it is updated before movement in each round.

\subsection{The Algorithm}
\label{Algo}
Based on the above discussions, Algorithm: \ref{alg:MaxCover} \textsc{MaxCover} is proposed. Each unstable node
executes the same algorithm by computing location and then moves in every round to cover maximum area over
an unbounded region. The algorithm terminates when all nodes become stable.

\begin{algorithm}[h]
\caption{\textsc{MaxCover({\it i})}}
\textbf{\underline{Initial round (or round 0) for unstable node $i$:}}
\begin{algorithmic}[1]
\IF{$i=0$} \STATE status $\leftarrow$  stable. 
\ENDIF
\IF{$i\neq 0$  \label{initialStart} and $j  = i~(mod~ 6)$} 
     \STATE {status $\leftarrow$ unstable}
     \STATE{ $\theta_i \leftarrow \frac{\pi}{3}\times j$, $min\_id \leftarrow j$}  \label{setId1} 
        \IF{$j=0$ } \STATE {$min\_id \leftarrow 6$} \label{setId2}
        \ENDIF
     \STATE{ $(X_i,Y_i) \leftarrow (\sqrt{3}r \cos\theta_i,\sqrt{3}r \sin\theta_i)$}
     \STATE{type $\leftarrow$ double}
     \STATE{$m\_id_i\leftarrow i$}
     \STATE {move to $(X_i,Y_i)$} 
 \ENDIF\label{initialEnd}\\
 \vspace*{2mm}
\textbf{\underline{$k$-th round for unstable nodes with type = double}}:
\IF{$m\_id_i==min\_id$} \STATE {status $\leftarrow$ stable} \ENDIF
\IF{$m\_id_i\neq min\_id$} \label{decomposeStart}
\STATE{$val_i \leftarrow  - \frac{1}{2} + \sqrt{\left(\frac{m\_id_i-j}{3}+\frac{1}{4}\right)}$}
\IF{$val_i$ is an integer}
\STATE{$\theta_i \leftarrow \theta_i+\frac{\pi}{3}$}
\STATE{type $\leftarrow$ single}
\STATE{$min\_id \leftarrow 3k(k+1)+j$}
\STATE{ $(X_i,Y_i) \leftarrow (X_i+ \sqrt{3}r \cos\theta_i,Y_i+ \sqrt{3}r \sin\theta_i)$}
\ELSE \STATE{ $(X_i,Y_i) \leftarrow (X_i+ \sqrt{3}r \cos\theta_i,Y_i+ \sqrt{3}r \sin\theta_i)$}
\STATE{$m\_id_i \leftarrow m\_id_i-6$}
\STATE{$min\_id \leftarrow 3k(k+1)+j$}
\ENDIF
\STATE {Move to $(X_i,Y_i)$} 
\ENDIF \label{decomposeEnd}\\
 \vspace*{2mm}
\textbf{\underline{$k$-th round for unstable nodes with type = single}}:
\IF{$m\_id_i=min\_id$}
\STATE{$status \leftarrow stable$}
\ELSE \STATE{$min\_id \leftarrow 3k(k+1)+j$}
\STATE{ $(X_i,Y_i) \leftarrow (X_i+ \sqrt{3}r \cos\theta_i,Y_i+ \sqrt{3}r \sin\theta_i)$}
\STATE {move to $(X_i,Y_i)$} 
\ENDIF
\end{algorithmic}
\label{alg:MaxCover}
\end{algorithm}

\section{Optimal Energy Consumption}\label{sec:ENERGY}
According to our model initially all nodes are deployed at a point and reach their final positions in several rounds for optimal spreading. Total energy consumption for optimal spreading would be minimum if each node moves straight from initial to final location. The straight movement is possible if all nodes compute their final positions before start their movement. In our proposed algorithm \ref{alg:MaxCover}, all sensor nodes in each round calculate positions based on their id's and move without exchanging any message. Therefore, it is possible to modify the algorithm \ref{alg:MaxCover} such that each node calculates its final destination position first, then move directly to the final position. The final position $(X_f, Y_f)$ of each node can be calculated iteratively using the same algorithm \ref{alg:MaxCover} without executing the move to $(X_i, Y_i)$ until satisfies the stability conditions: $m\_id_i = min\_id$.

\section{Arbitrary sensor deployment}
If initial deployment of all sensor nodes are random over the plane instead of a point unlike our proposed model,
optimal spreading is still possible with following modification of the algorithm \ref{alg:MaxCover}.
After random deployment, node with id 0 informs its position $(x_0,y_0)$ to all other nodes, based on the position of the node with id 0, all other nodes compute their final destination positions locally and finally move to the respective position in one round. Each node can calculate their final position iteratively as explained above.
Optimal energy consumption may not be possible for the deployment of sensor nodes randomly over a plane.

\section{simulation result}
We have written a C++ program to simulate our algorithm for different number of sensor nodes with initial deployment at origin. The output of the program is the final coordinates of the nodes. With this data, we have plotted the corresponding sensing disks which shows the optimal spreading. The following Fig. \ref{fig:simulate1}, Fig. \ref{fig:simulate2} and Fig. \ref{fig:simulate3}, are the output of our simulation results for different input for number of sensor nodes. The node with id 0 become stable in the initial position of deployment and the sensing disk corresponding to the node is shown by the dotted circle.
\begin{figure}[h]
  \centering
  \includegraphics[width=.3\textwidth]{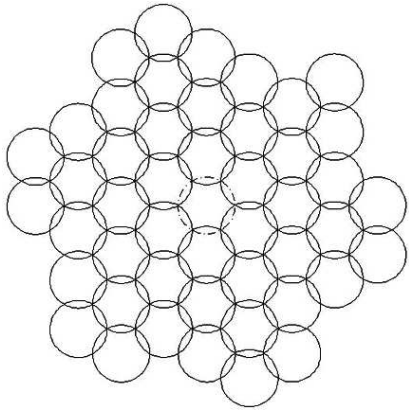}
  \caption{optimal spreading for $n=48$ nodes}
  \label{fig:simulate1}
\end{figure}
\begin{figure}[h]
  \centering
  \includegraphics[width=.3\textwidth]{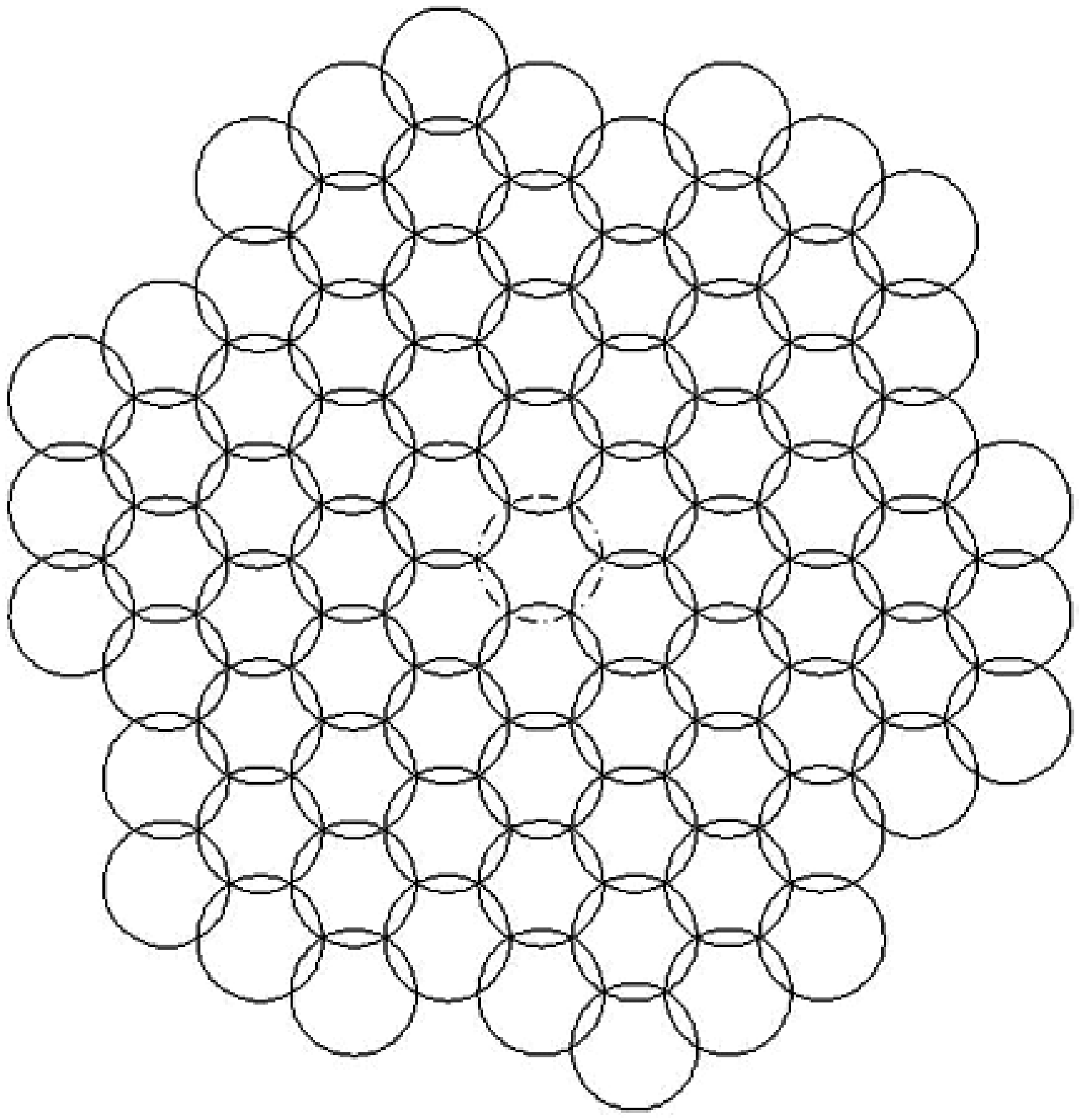}
  \caption{optimal spreading for $n=79$ nodes}
  \label{fig:simulate2}
\end{figure}
\begin{figure}[h]
 \centering
 \includegraphics[width=.3\textwidth]{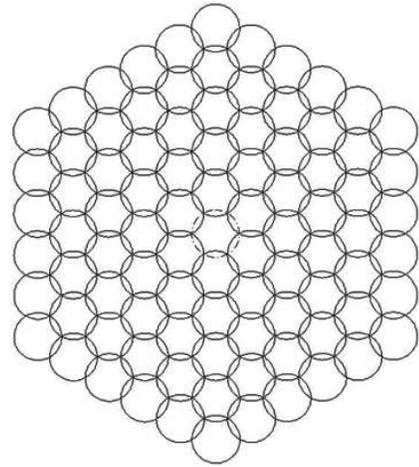}
 \caption{optimal spreading for n=91 nodes}
 \label{fig:simulate3}
\end{figure}

\section{Conclusion}
\label{concl}
In this paper we have presented a distributed synchronous algorithm: \textsc{MaxCover}
for maximum spreading without creating any coverage hole on an unbounded region.
The movement of each sensor node only depends on its unique {\it id} and present position.
No message exchange is required among the sensor nodes for execution of the algorithm. The number of rounds required for optimal spreading is $O(\sqrt{n})$.

Though in this paper we have assumed that all the sensor nodes are geometric points on a plane and initially all nodes are deployed at a point but same algorithm is also applicable for real (non-point) mobile sensors with following minor modifications of initial deployment. The nodes should be deployed in a place in such a way that they can form a completely connected graph with respect to their transmission range and they should have the knowledge of the position of the node with $id$ 0. It may be possible to have either prior knowledge of the location of node with $id$ 0 or after deployment the node with $id$ 0 informs its location to all other nodes by a broadcast. In each round, all unstable nodes move, either at the calculated location or any adjacent location of the calculated location.  But during calculation of the next location, all unstable nodes should use calculated location of the previous round. All nodes which are allowed to be stable in any round should finally move to their calculated locations.

We have explained how proposed algorithm can be modified to achieve optimal energy consumption and how the algorithm can handle random deployment of sensor nodes. In future we will try to focus on the covering in presence of boundary and obstacles.

\bibliographystyle{plain}
\bibliography{bib}
\end{document}